\documentclass{llncs}
\usepackage[utf8]{inputenc}
\usepackage{amssymb,amsmath}
\usepackage{graphicx, enumerate}
\usepackage{xspace}
\usepackage{subfigure}
\usepackage{hyperref}
\hypersetup{
  colorlinks,%
  citecolor=black,%
  filecolor=black,%
  linkcolor=black,%
  urlcolor=black
  }

\usepackage{tikz}
\usetikzlibrary{decorations.pathmorphing,patterns,decorations.pathreplacing,arrows,calc}

\usepackage[inline,nomargin,draft]{fixme}
\fxsetface{inline}{\bfseries\color{red}}

\spnewtheorem{myclaim}{Claim}{\bfseries}{\itshape}
\spnewtheorem{fact}[theorem]{Fact}{\bfseries}{\itshape}
\spnewtheorem{observation}[theorem]{Observation}{\bfseries}{\itshape}

\newcommand{\ceil}[1]{\left\lceil #1 \right\rceil}
\newcommand{\Oh}{\mathcal{O}}
\newcommand{\List}{\mathcal{L}}
\newcommand{\Program}{\mathcal{P}}
\newcommand{\lcs}{LCS}
\newcommand{\eps}{\varepsilon}
\renewcommand{\epsilon}{\varepsilon}
\renewcommand{\r}{r}

\makeatletter
\newenvironment{oneshot}[1]{\@begintheorem{#1.}{\unskip}}{\@endtheorem}
\makeatother

\DeclareMathOperator*{\per}{\ensuremath{per}}

\begin{document}
\title{Sublinear Space Algorithms for\\ the Longest Common Substring Problem}
\author{
Tomasz Kociumaka\inst{1}\fnmsep\thanks{
Supported by Polish budget funds for science in 2013-2017 as a research project under the `Diamond Grant' program
}
\and
Tatiana Starikovskaya\inst{2}\fnmsep\thanks{
Partly supported by Dynasty Foundation.
}\and
Hjalte~Wedel~Vildh{\o}j\inst{3}
}

\institute{
Institute of Informatics, University of Warsaw
\and
National Research University Higher School of Economics (HSE)
\and
Technical University of Denmark, DTU Compute
}

\date{}

\maketitle

\begin{abstract}
Given $m$ documents of total length $n$, we consider the problem of finding a longest string common to at least $d \geq 2$ of the documents. This problem is known as the \emph{longest common substring (LCS) problem} and has a classic $\Oh(n)$ space and $\Oh(n)$ time solution (Weiner [FOCS'73], Hui [CPM'92]). However, the use of linear space is impractical in many applications. In this paper we show that for any trade-off parameter $1 \leq \tau \leq n$, the LCS problem can be solved in $\Oh(\tau)$ space and $\Oh(n^2/\tau)$ time, thus providing the first smooth deterministic time-space trade-off from constant to linear space. The result uses a new and very simple algorithm, which computes a $\tau$-additive approximation to the LCS in $\Oh(n^2/\tau)$ time and $\Oh(1)$ space. We also show a time-space trade-off lower bound for deterministic branching programs, which implies that any deterministic RAM algorithm solving the LCS problem on documents from a sufficiently large alphabet in $\Oh(\tau)$ space must use $\Omega(n\sqrt{\log(n/(\tau\log n))/\log\log(n/(\tau\log n)})$ time.
\end{abstract}

\section{Introduction}
The  \emph{longest common substring (LCS) problem} is a fundamental and classic string problem with numerous applications.
Given $m$ strings $T_1,T_2,\ldots,T_m$ (the documents) from an alphabet $\Sigma$ and a parameter $2 \leq d \leq m$, the LCS problem is to compute a longest string occurring in least $d$ of the $m$ documents. We denote such a string by $\lcs$ and use $n=\sum_{i=1}^m |T_i|$ to refer to the total length of the documents.

The classic text-book solution to this problem is to build the (generalized) suffix tree of the documents and find the node that corresponds to $\lcs$~\cite{Colors,WeinerSTconstruction,Gusfield}. While this can be achieved in linear time, it comes at the cost of using $\Omega(n)$ space\footnote{Throughout the paper, we measure space as the number of words in the standard unit-cost word-RAM model with word size $w = \Theta(\log n)$ bits.} to store the suffix tree. In applications with large amounts of data or strict space constraints, this renders the classic solution impractical. A recent example of this challenge is automatic generation of signatures for identifying zero-day worms by solving the LCS problem on internet packet data~\cite{afek2013automated,kreibich2004honeycomb,wang2006anomalous}. The same challenge is faced if the length of the longest common substring is used as a measure for plagiarism detection in large document collections.

To overcome the space challenge of suffix trees, succinct and compressed data structures have been subject to extensive research~\cite{grossi2005compressed,navarro2007compressed}. Nevertheless, these data structures still use $\Omega(n)$ bits of space in the worst-case, and are thus not capable of providing truly sublinear space solutions to the LCS problem.

\subsection{Our Results}
We give new sublinear space algorithms for the LCS problem. They are designed for the word-RAM model with word size $w=\Omega(\log n)$, and work for integer alphabets $\Sigma = \{1,2,\ldots,\sigma\}$ with $\sigma = n^{\Oh(1)}$. Throughout the paper, we regard the output to the LCS problem as a pair of integers referring to a substring in the input documents, and thus the output fits in $\Oh(1)$ machine words.

As a stepping stone to our main result, we first show that an additive approximation of $\lcs$ can be computed in constant space. We use $|LCS|$ to denote the length of the longest common substring.

\begin{theorem}\label{thm:approximation}
There is an algorithm that given a parameter $\tau$, $1\le \tau \le n$, runs in $\Oh(n^2/\tau)$ time and $\Oh(1)$ space,
and outputs a string, which is common to at least $d$ documents and has length at least $|LCS|-\tau+1$.
\end{theorem}
The solution is very simple and essentially only relies on a constant space pattern matching algorithm as a black-box. 
We expect that it could be of interest in applications where an approximation of $\lcs$ suffices.

For $\tau=1$ we obtain the corollary:

\begin{corollary}
$\lcs$ can be computed in $\Oh(1)$ space and $\Oh(n^2)$ time.
\end{corollary}
To the best of our knowledge, this is the first constant space $\Oh(n^2)$-time algorithm for the LCS problem. Given that it is a simple application of a constant space pattern matching algorithm, it is an interesting result on its own.

Using \autoref{thm:approximation} we are able to establish our main result, which gives the first deterministic time-space trade-off from constant to linear space:

\begin{theorem}\label{thm:upperbound}
There is an algorithm that given a parameter $\tau$, $1\le \tau \le n$, computes $\lcs$ in $\Oh(\tau)$ space and $\Oh(n^2/\tau)$ time.
\end{theorem}
Previously, no deterministic trade-off was known except in the restricted setting of $n^{2/3} < \tau \le n$, where two of the authors showed that the problem allows an $\Oh((n^2/\tau) d \log^2 n( \log^2 n +d))$-time and $\Oh(\tau)$-space trade-off~\cite{LCS2013}. Our new solution is also strictly better than the $\Oh((n^2/\tau) \log n)$-time and $\Oh(\tau)$-space randomized trade-off, which correctly outputs $\lcs$ with high probability (see \cite{LCS2013} for a description).

Finally, we prove a time-space trade-off lower bound for the LCS problem over large-enough alphabets,
which remains valid even restricted to two documents.

\begin{theorem}\label{thm:lowerbound}
Given two documents of total length $n$ from an alphabet $\Sigma$ of size at least $n^2$, any deterministic RAM algorithm, which uses $\tau\le \frac{n}{\log n}$ space to compute the longest common substring of both documents, must use time $\Omega (n \sqrt{\log (n/(\tau\log n))/ \log\log (n/(\tau \log n))})$.
\end{theorem}
We prove the bound for non-uniform deterministic branching programs, which are known to simulate deterministic RAM algorithms with constant overhead.
The lower bound of \autoref{thm:lowerbound} implies that the classic linear-time solution is close to asymptotically optimal in the sense that there is no hope for a linear-time and $o(n/ \log n)$-space algorithm that solves the LCS problem on polynomial-sized alphabets.


\section{Upper Bounds}
Let $T$ be a string of length $n > 0$. Throughout the paper, we use the notation $T[i..j]$, $1 \leq i \leq j \leq n$, to denote the substring of $T$ starting at position $i$ and ending at position $j$ (both inclusive). We use the shorthand $T[..i]$ and $T[i..]$ to denote $T[1..i]$ and $T[i..n]$ respectively. 

A suffix tree of $T$ is a compacted trie on suffixes of $T$ appended with a unique letter (sentinel) $ \$$ to guarantee one-to-one correspondence between suffixes and leaves of the tree. The suffix tree occupies linear space. Moreover, if the size of the alphabet is polynomial in the length of $T$, then the suffix tree can be constructed in linear time~\cite{DBLP:conf/focs/Farach97}. We refer to nodes of the suffix tree as \emph{explicit nodes}, and to nodes of the underlying trie, which are not preserved in the suffix tree, as \emph{implicit nodes}. Note that each substring of $T$ 
corresponds to a unique explicit or implicit node, the latter can be specified by the edge it belongs to and its distance to the upper endpoint of the edge.

A generalized suffix tree of strings $T_1$, $T_2, \ldots, T_m$ is a trie on all suffixes of these strings appended with sentinels $ \$_i$. It occupies  linear space and for polynomial-sized alphabets can also be constructed in linear time.

\paragraph{Classic solution.}
As a warm-up, we briefly recall how to solve the LCS problem in linear time and space. Consider the generalized suffix tree of the documents $T_1, T_2, \ldots, T_m$, where leaves corresponding to suffixes of $T_i$, $i=1,2, \ldots,m$, are painted with color $i$. 
%
%
The main observation is that $\lcs$ is the label of a deepest explicit node with leaves of at least $d$ distinct colors in its subtree. Hui~\cite{Colors} showed that given a tree with $\Oh(n)$ nodes where some leaves are colored, it is possible to compute the number of distinctly colored leaves below all nodes in $\Oh(n)$ time. Consequently, we can locate the node corresponding to $\lcs$ in $\Oh(n)$ time and $\Oh(n)$ space.

\subsection{Approximating \lcs\ in Constant Space}
Given a pattern and a string, it is possible to find all occurrences of the pattern in the string using constant space and linear time (see~\cite{DBLP:journals/tcs/BreslauerGM13} and references therein). 
We use this result in the following $\Oh(1)$-space additive approximation algorithm.

\begin{lemma}\label{lem:one}
There is an algorithm that given integer parameters $\ell$, $r$ satisfying $1\le \ell < r \le n$,
runs in $\Oh\big(\frac{n^2}{r-\ell}\big)$ time and constant space, and returns NO if $|\lcs| < \ell$, YES if $|\lcs| \ge r$, and an arbitrary answer otherwise.
\end{lemma}
\begin{proof}
Let $S = T_1\$_1 T_2\$_2 \ldots T_m\$_m$ and $\tau = r-\ell$. Consider substrings $S_k = S[k\tau+1..k\tau+\ell]$ for $k=0,\ldots,\big\lfloor\frac{|S|}{\tau}\big\rfloor$.
For each $S_k$ we use a constant-space pattern matching algorithm to count the number of documents $T_i$ containing an occurrence of $S_k$. We return YES if for any $S_k$ this value is at least $d$ and NO otherwise.

If $|\lcs| < \ell$, then any substring of $S$ of length $\ell$~--- in particular, any $S_k$~--- occurs in less than $d$ documents. Consequently, in this case the algorithm will return NO. On the other hand, any substring of $S$ of length $r$ contains some $S_k$, so if $|\lcs| \ge r$, then some $S_k$ occurs in at least $d$ documents, and in this case the algorithm will return YES.
\qed
\end{proof}
To establish \autoref{thm:approximation} we perform a ternary search using \autoref{lem:one} with the modification that if the algorithm
returns YES, it also outputs a string of length $\ell$ common to at least $d$ documents.
We maintain an interval $R$ containing $|\lcs|$; initially $R=[1, n]$. In each step we  set $\ell$ and $r$ (approximately) in $1/3$ and $2/3$ of $R$, so that we can reduce $R$ by $\lfloor{|R|/3}\rfloor$. We stop when $|R| \leq \tau$. The time complexity bound forms a geometric progression dominated by the last term, which is $\Oh(n^2/\tau)$. This concludes the proof of the following result.

\begin{oneshot}{\autoref{thm:approximation}}
There is an algorithm that given a parameter $\tau$, $1\le \tau \le n$, runs in $\Oh(n^2/\tau)$ time and $\Oh(1)$ space,
and outputs a string, which is common to at least $d$ documents and has length at least $|LCS|-\tau+1$.
\end{oneshot}

\subsection{An $\Oh(\tau)$-Space and $\Oh(n^2/\tau)$-Time Solution}
We now return to the main goal of this section. Using \autoref{thm:approximation}, we can assume to know $\ell$ such that $\ell \le |\lcs| < \ell+\tau$. Organization of the text below is as follows. 
First, we explain how to compute $\lcs$ if $\ell = 1$. 
Then we extend our solution so that it works with larger values of $\ell$.
Here we additionally assume that the alphabet size is constant and later, in \autoref{ss:large}, we remove this assumption. 

\subsubsection*{Case $\ell = 1$.}
From the documents $T_1, T_2, \ldots, T_m$ we compose two lists of strings. First, we consider ``short'' documents $T_j$ with $|T_j| < \tau$. 
We split them into groups of total length in $[\tau, 1+2\tau]$ (except for the last group, possibly). 
For each group we add a concatenation of the documents in this group, appended with sentinels $ \$_j$, to a list $\List_1$. 
Separately, we consider ``long'' documents $T_j$ with $|T_j| \ge \tau$. 
For each of them we add to a list $\List_2$ its substrings starting at positions of the form $k\tau+1$ for integer $k$ and in total covering $T_j$. 
These substrings are chosen to have length $2\tau$, except for the last whose length is in $[\tau, 2\tau]$. 
We assume that substrings of the same document $T_j$ occur contiguously in $\List_2$ and append them with $ \$_j$. 
The lists $\List_1$ and $\List_2$ will not be stored explicitly but will be generated on the fly while scanning the input. 
Note that $|\List_1 \cup \List_2| = \Oh(n/\tau)$. 

\begin{observation}
Since the length of $\lcs$ is between $1$ and $\tau$, $\lcs$ is a substring of some string $S_k \in \List_1 \cup \List_2$. Moreover, it is a label of an explicit node of the suffix tree of $S_k$ or of a node where a suffix of some $S_i\in \List_1\cup\List_2 $ branches out of the suffix tree of $S_k$.
\end{observation}
We process candidate substrings in groups of $\tau$, using the two lemmas below.

\begin{lemma}\label{lem:tau-candidates-1}
Consider a suffix tree of $S_k$ with $\tau$ marked nodes (explicit or implicit). There is an $\Oh(n)$-time and $\Oh(\tau)$-space algorithm that counts the number of short documents containing an occurrence of the label of each marked node.
\end{lemma}
\begin{proof}
For each marked node we maintain a counter $c(v)$ storing the number of short documents the label of $v$ occurs in. 
Counters are initialized with zeros. We add each string $S_i \in \List_1$ to the suffix tree of $S_k$ in $\Oh(\tau)$ time. 
By \emph{adding} a string to the suffix tree of another string, we mean constructing the generalized suffix tree of both strings and establishing pointers from explicit nodes of the generalized suffix tree to the corresponding nodes of the original suffix tree. 
We then paint leaves representing suffixes of $S_i$: namely, we paint a leaf with color $j$ if the corresponding suffix of $S_i$ starts within a document $T_j$ (remember that $S_i$ is a concatenation of short documents). 
Then the label of a marked node occurs in $T_j$ iff this node has a leaf of color $j$ in its subtree. 
Using Hui's algorithm we compute the number of distinctly colored leaves in the subtree of each marked node $v$ and add this number to $c(v)$. 
After updating the counters, we remove colors and newly added nodes from the tree.
Since all sentinels in the strings in $\List_1$ are distinct, the algorithm is correct. It runs in $\Oh(|\List_1| \tau+\tau) = \Oh(n)$ time.
\qed
\end{proof}

\begin{lemma}\label{lem:tau-candidates-2}
Consider a suffix tree of $S_k$ with $\tau$ marked nodes (explicit or implicit). 
There is an $\Oh(n)$-time and $\Oh(\tau)$-space algorithm that counts the number of long documents containing an occurrence of the label of each marked node.
\end{lemma}
\begin{proof}
For each of the marked nodes we maintain a variable $c(v)$ counting the documents where the label of $v$ occurs.
A single document might correspond to several strings $S_i$, so
we also keep an additional variable $m(v)$, which prevents increasing $c(v)$ several times for a single document.
As in~\autoref{lem:tau-candidates-1}, we add each string $S_i \in \List_2$ to the suffix tree of $S_k$. 
For each marked node $v$ whose subtree contains a suffix of $S_i$ ending with $ \$_j$, we compare $m(v)$ with $j$. 
We increase $c(v)$ only if $m(v)\ne  j$, also setting $m(v) = j$ to prevent further increases for the same document.
Since strings corresponding to the $T_j$ occur contiguously in $\List_2$, the algorithm is correct. Its running time is $\Oh(|\List_2| \tau+\tau) = \Oh(n)$.
\qed
\end{proof}

Let $\List = \List_1 \cup \List_2$. If $\lcs$ is a substring of $S_k \in \mathcal{L}$, we can find it as follows: we construct the suffix tree of $S_k$, mark its explicit nodes and nodes where suffixes of $S_i \in \List$ ($i\neq k$) branch out, and determine the deepest of them which occurs in at least $d$ documents. 
Repeating for all $S_k \in \List$, this allows us to determine $\lcs$. 
To reduce the space usage to $\Oh(\tau)$, we use \autoref{lem:tau-candidates-1} and \autoref{lem:tau-candidates-2} for batches of $\Oh(\tau)$ marked nodes in the suffix tree of $S_k$ at a time. 
Labels of all marked node are also labels of explicit nodes in the generalized suffix tree of $T_1,\ldots,T_m$.
In order to achieve good running time we will make sure that marked nodes have, over all $S_k\in \List$, distinct labels. This will imply that we use \autoref{lem:tau-candidates-1} and \autoref{lem:tau-candidates-2} only $\Oh(n/\tau)$ times, and hence spend $\Oh(n^2/\tau)$ time overall.

We consider each of the substrings $S_k \in \List$ in order. 
We start by constructing a suffix tree for $S_k$. 
To make sure the labels of marked nodes are distinct, we shall exclude some (explicit and implicit) nodes of $S_k$. 
Each node is going to be excluded together with all its ancestors or descendants, so that it is easy to test whether a particular node is excluded. 
(It suffices to remember the highest and the lowest non-excluded node on each edge, if any,  $\Oh(\tau)$ nodes in total.)

First of all, we do not need to consider substrings of $S_1,\ldots,S_{k-1}$. Therefore we add each of strings $S_1, S_2, \ldots, S_{k-1}$ to the suffix tree (one by one) and exclude nodes common to $S_k$ and these strings from consideration. Note that in this case a node is excluded with all its ancestors.

Then we consider all strings $S_k, S_{k+1}, S_{k+2}, \ldots$ in turn. 
For each string we construct the generalized suffix tree of $S_k$ and the current $S_i$ and iterate over explicit nodes of the tree whose labels are substrings of $S_k$. 
If a node has not been excluded, we mark it. 
Once we have $\tau$ marked nodes (and if any marked nodes are left at the end), we apply \autoref{lem:tau-candidates-1} and \autoref{lem:tau-candidates-2}. 
If the label of a marked node occurs in at least $d$ documents, then we can exclude the marked node and all its ancestors. Otherwise, we can exclude it with all its descendants.

Recall that $\lcs$ is a label of one of the explicit inner nodes of the generalized suffix tree of $T_1, T_2, \ldots, T_m$, i.e., there are $\Oh(n)$ possible candidates for~$\lcs$. 
Moreover, we are only interested in candidates of length at most $\tau$, and each such candidate corresponds to an explicit node of the generalized suffix tree of a pair of strings from $\List$. 
The algorithm process each such candidate exactly once due to node exclusion. 
Thus, its running time is $\Oh(\frac{n}{\tau} n+\tau) = \Oh(n^2 / \tau)$. 
At any moment it uses $\Oh(\tau)$ space.

\subsubsection*{General case.}
If $\ell < 10 \tau$ we can still use the technique above, adjusting the multiplicative constants
in the complexity bounds. Thus, we can assume $\ell > 10 \tau$.

Documents shorter than $\ell$ cannot contain $\lcs$ and we ignore them. For each of the remaining documents $T_j$ we add to a list $\List$ its substrings starting at positions of the form $k\tau+1$ for integer $k$ and in total covering $T_j$. The substrings are chosen to have length $\ell + 2\tau$, except for the last whose length is in the interval $[\ell, \ell+2\tau]$. Each substring is appended with $ \$_j$, and we assume that the substrings of the same document occur contiguously.

\begin{observation}\label{obs:todo}
Since the length of $\lcs$ is between $\ell$ and $\ell+\tau$, $\lcs$ is a substring of some string $S_k \in \List$. Moreover, it is the label of a node of the suffix tree of $S_k$ where a suffix of another string $S_i \in \List$ branches out. (We do not need to consider explicit nodes of the suffix tree as there are no short documents.)
\end{observation}

As before, we consider strings $S_k\in \List$ in order and check all candidates which are substrings of $S_k$ but not any $S_i$ for $i<k$. However, in order to make the algorithm efficient, we replace all strings $S_i$, including $S_k$, with strings $\r_k(S_i)$, each of length $\Oh(\tau)$. To define the mapping $\r_k$ we first introduce some necessary notions.

We say that $S[1..p]$ is a period of a string $S$ if $S[i]=S[i+p]$, $1 \le i \le |S|-p$. The length of the shortest period of $S$ is denoted as $\per(S)$. We say that a string $S$ is \emph{primitive} if its shortest period is not a proper divisor of $|S|$. Note that $\rho = S[1..\per(S)]$ is primitive and therefore satisfies the following lemma:
\begin{lemma}[Primitivity Lemma \cite{AlgorithmsOnStrings}]\label{lem:sync}
Let $\rho$ be a primitive string. Then $\rho$ has exactly two occurrences in a string $\rho\rho$.
\end{lemma}

Let $Q_k = S_k[1+2\tau .. \ell]$; note that $|Q_k|=\ell-2\tau\ge 8\tau$.
Let $\per(Q_k)$ be the length of the shortest period $\rho$ of $Q$. 
 If $\per(Q_k) > 4\tau$, we define $Q_k' = \#$, where $\#$ is a special letter that does not belong to the main alphabet. 
Otherwise $Q_k$ can be represented as $\rho^t\rho'$, where $\rho'$ is a prefix of $\rho$. 
We set $Q_k'=\rho^{t'}\rho'$ for $t'\le t$ chosen so that $8\tau\le |Q_k'|<12\tau$.
For any string $S$ we define $\r_k(S)=\eps$ if $S$ does not contain $Q_k$, and a string obtained from $S$ 
by replacing the first occurrence of $Q_k$ with $Q'_k$ otherwise. Below we explain how to compute $Q'_k$.

\begin{lemma}
One can decide in linear time and constant space if $\per(Q_k) \leq 4 \tau$ and provided that this condition holds, compute $\per(Q_k)$.
\end{lemma}
\begin{proof}
Let $P$ be the prefix of $Q_k$ of length $\ceil{|Q_k|/2}$ and $p$ be the starting position of the second occurrence of $P$ in $Q_k$, if any. The position $p$ can be found in $\Oh(|Q_k|)$ time by a constant-space pattern matching algorithm.

We claim that if $\per(Q_k) \leq 4\tau \leq \ceil{|Q_k|/2}$, then $p = \per(Q_k)+1 $. Observe first that in this case $P$ occurs at a position $per(Q_k)+1$, and hence $p \le \per(Q_k)+1$. Furthermore, $p$ cannot be smaller than $\per(Q_k)+1$, because otherwise $\rho = Q_k[1..\per(Q_k)]$ would occur in $\rho\rho=Q_k[1..2\per(Q_k)]$ at the position $p$. The shortest period $\rho$ is primitive, so this is a contradiction with \autoref{lem:sync}. 

The algorithm compares $p$ and $4 \tau + 1$. If $p \le  4 \tau + 1$, it uses letter-by-letter comparison to determine whether $Q_k[1..p-1]$ is a period of $Q_k$. If so, by the discussion above $\per(Q_k)=p-1$, and the algorithm returns it. 
Otherwise $\per(Q_k)>4\tau$.
The algorithm runs in $\Oh(|Q_k|)$ time and uses constant space.
\qed
\end{proof}

\begin{fact}\label{fct:replace}
Suppose that a string $S$, $|S| \le |Q_k|+4\tau$, contains $Q_k$ as a substring.
Then
\begin{enumerate}[(a)]\parsep=0pt
  \item\label{it:there} replacing with $Q_k'$ any occurrence of $Q_k$ in $S$ results in $\r_k(S)$,
  \item\label{it:back} replacing with $Q_k$ any occurrence of $Q'_k$ in $\r_k(S)$ results in $S$.
\end{enumerate}
\end{fact}
\begin{proof}
We start with (\ref{it:there}).
Let $i$ and $i'$ be the positions of the first and last occurrence of $Q_k$ in $S$.
We have $1 \le i \le i'\le |S|-|Q_k|+1$, so $i'-i \le |S|-|Q_k|\le 4\tau$.
If $\per(Q_k) > 4\tau$ this implies that $i' - i = 0$, or, in other words,
that $Q_k$ has just one occurrence in $S$.

On the other hand, if $\per(Q_k) \le 4\tau$, we observe that $i'-i \le 4\tau = 8\tau - 4\tau \le |Q_k|-\per(Q_k)$.
Therefore the string $\rho=S[i'..i'+\per(Q_k)-1]$ fits within $Q_k=S[i..i+|Q_k|-1]$. It is primitive and
\autoref{lem:sync} implies that $\rho$ occurs in $\rho^t\rho'$ only $t$ times,
so $i'=i+j \cdot per(Q_k)$ for some integer $j\le t$.
Therefore all occurrences of $Q_k$ lie in the substring of $S$ of the form $\rho^s\rho'$ for some $s\ge t$.
Thus, replacing any of these occurrences with $Q'_k$ leads to the same result, $\r_k(S)$.

Now, let us prove (\ref{it:back}).
Note that if we replace an occurrence of $Q'_k$ in $\r_k(S)$ with $Q_k$,
by (\ref{it:there}) we obtain a string $S'$ such that $\r_k(S')=\r_k(S)$. Moreover all such strings $S'$
can be obtained by replacing some occurrence of $Q'_k$, in particular this is true for~$S$.

If $\per(Q_k) > 4\tau$, since $\#$ does not belong to the main alphabet, $Q'_k$ has exactly one occurrence in $\r_k(S)$
and the statement holds trivially. For the other case we proceed as in the proof of (\ref{it:there})
showing that all occurrences of $Q'_k$ are in fact substrings of a longer substring of $S$ of the form $\rho^{s'}\rho'$
for some $s'\ge t'$.\qed
\end{proof}

\begin{lemma}\label{lem:replace}
Consider strings $P$ and $S$, such that $|S|\le |Q_k|+4\tau$ and $P$ contains $Q_k$ as a substring.
Then $P$ occurs in $S$ at position $p$ if and only if $\r_k(P)$ occurs in $\r_k(S)$ at position $p$. 
\end{lemma}
\begin{proof}
First, assume that $P$ occurs in $S$ at a position $p$. 
This induces an occurrence of $Q_k$ in $S$ within the occurrence of $P$, 
and replacing this occurrence of $Q_k$ with $Q'_k$ gives $\r_k(S)$ by \autoref{fct:replace}(\ref{it:there}).
This replacement also turns the occurrence of $P$ at the position $p$ into an occurrence of $\r_k(P)$.

Now, assume $\r_k(P)$ occurs in $\r_k(S)$ at the position $p$.
Since $\r_k(P) \ne \eps$, this means that $\r_k(S) \ne \eps$ and that
 $Q'_k$ occurs in $\r_k(S)$ (within the occurrence of $\r_k(P)$).
By \autoref{fct:replace}(\ref{it:back}) replacing this occurrence of $Q'_k$ with $Q_k$ 
turns $\r_k(S)$ into $S$ and the occurrence of $\r_k(P)$ at the position $p$
into an occurrence of $P$.
\qed
\end{proof}

Observe that applied for $S=S_k$, \autoref{lem:replace} implies that $\r_k$
gives a bijection between substrings of $S_k$ of length $\ge \ell=|Q_k|+2\tau$
and substrings of $\r_k(S_k)$ of length $\ge |Q'_k|+2\tau$.
Moreover, it shows that any substring of $S_k$
of length $\ge \ell$ occurs in $S_i$ iff the corresponding substring of $\r_k(S_k)$ occurs in $\r_k(S_i)$.

This lets us apply the technique described in the previous section to find $\lcs$ provided that it occurs in $S_k$
but not $S_i$ with $i<k$. Strings $\r_k(S_i)$ are computed in parallel with a constant-space pattern matching algorithm for a pattern $Q_k$ in the documents of length $\ell$ or more, which takes $\Oh(n)$ time in total. The list $\List$ is composed $r_k(S_i)$ obtained from long documents, and we use \autoref{lem:tau-candidates-2} to compute the number of documents each candidates occurs in.
%
%

Compared to the arguments of the previous section, we additionally exclude nodes of depth less than $|Q'_k|+2\tau$ to make sure that each marked node is indeed $\r_k(P)$ for some substring $P$ of $S_k$ of length at least $\ell=|Q_k|+2\tau$.
This lets us use the amortization by the number of explicit nodes in the generalized suffix tree of $T_1,\ldots,T_m$.
More precisely, if a node with label $\r_k(P)$ is marked, we charge $P$, which is guaranteed to be explicit in the generalized suffix tree.
This implies $\Oh(n^2 / \tau)$-time and $\Oh(\tau)$-space bounds.

\subsection{Large alphabets}\label{ss:large}
In this section we describe how to adapt our solution so that it works for alphabets of size $n^{\Oh(1)}$.
Note that we have used the constant-alphabet assumption only to make sure that suffix trees can be efficiently constructed. If the alphabet is not constant,
a suffix tree of a string can be constructed in linear time plus the time of sorting its letters~\cite{DBLP:conf/focs/Farach97}.
If $\tau > \sqrt{n}$, the size of the alphabet is $n^{\Oh(1)} = \tau^{\Oh(1)}$ and hence any suffix tree used by the algorithm can be constructed in $\Oh(\tau)$ time.

Suppose now that $\tau \le \sqrt{n}$ and $\ell=1$. Our algorithm uses suffix trees in a specific pattern: in a single phase it builds the suffix tree of $S_k$
and then constructs the generalized suffix tree of $S_k$ and $S_i$ for each~$i$. Note that the algorithm only needs information about the nodes of the suffix tree of $S_k$, the nodes where suffixes of $S_i \in \mathcal{L}$ branch out, and leaves of the generalized suffix tree.
None of these changes if we replace each letter of $\Sigma$ occurring in $S_i$, but not in $S_k$, with a special letter which does not belong to $\Sigma$.

Thus our approach is as follows: first we build a deterministic dictionary, mapping letters of $S_k$ to integers of magnitude $\Oh(|S_k|)=\Oh(\tau)$
and any other letter of the main alphabet to the special letter. The dictionary can be constructed in $\Oh(\tau \log^2\log\tau)$ time~\cite{DBLP:conf/icalp/Ruzic08,DBLP:journals/jal/Han04}.
Then instead of building the generalized suffix tree of $S_k$ and $S_i$ we build it for the corresponding strings with letters mapped using the dictionary. 
In general, when $\ell$  is large, we apply the same idea with $r_k(S_k)$ and $r_k(S_i)$ instead of $S_k$ and $S_i$ respectively. 

In total, the running time is $\Oh(n^2/\tau + n\log^2\log \tau)$. For $\tau \le \sqrt{n}$ the first
term dominates the other, i.e. we obtain an $\Oh(n^2/\tau)$-time solution.

\begin{oneshot}{\autoref{thm:upperbound}}
There is an algorithm that given a parameter $\tau$, $1\le \tau \le n$, computes $\lcs$ in $\Oh(n^2/\tau)$ time
using $\Oh(\tau)$ space.
\end{oneshot}

\section{A Time-Space Trade-Off Lower Bound}
Given $n$ elements over a domain $D$, the \emph{element distinctness problem} is to decide whether all $n$ elements are distinct. Beame~et~al.~\cite{DecisionProblemsRandomized} showed that if $|D| \ge n^2$, then any RAM algorithm solving the element distinctness problem in $\tau$ space, must use at least $\Omega (n \sqrt{\log (n/(\tau \log n))/ \log\log (n/(\tau \log n))})$ time.\footnote[1]{Note that in \cite{DecisionProblemsRandomized,DBLP:journals/siamcomp/BorodinC82} the space consumption is measured in bits. The version of RAM used there is unit-cost with respect to time and log-cost with respect to space.}

The element distinctness (ED) problem can be seen as a special case of the LCS problem where we have $m=n$ documents of length $1$ and want to find the longest string common to at least $d=2$ documents. Thus, the lower bound for ED also holds for this rather artificial case of the LCS problem. Below we show that the same bound holds with just $m=2$ documents. The main idea is to show an analogous bound for a two-dimensional variant of the element distinctness problem, which we call the \emph{element bidistinctness problem}. The LCS problem on two documents naturally captures this problem. The steps are similar to those for the ED lower bound by Beame~et~al. \cite{DecisionProblemsRandomized}, but the details differ. We start by introducing the necessary definitions of branching programs and embedded rectangles. We refer to \cite{DecisionProblemsRandomized} for a thorough overview of this proof technique.


\paragraph{Branching Programs.}
A $n$-variate branching program $\Program$ over domain $D$ is an acyclic directed graph with the following properties:
(1) there is a unique source node denoted $s$, (2) there are two sink nodes, one labelled by $0$ and one labelled by $1$, (3) each nonsink node $v$ is assigned an index $i(v) \in [1,n]$ of a variable, and (4) there are exactly $|D|$ arcs out of each nonsink node, labelled by distinct elements of~$D$.
A branching program is executed on an input $x \in D^n$ by starting at $s$, reading the variable $x_{i (s)}$ and following the unique arc labelled by $x_{i(s)}$. This process is continued until a sink is reached and the output of the computation is the label of the sink. For a branching program $\Program$, we define its \emph{size} as the number of nodes, and its \emph{length} as the length of the longest path from $s$ to a sink node.

\begin{lemma}[see page~2 of \cite{DBLP:journals/siamcomp/BorodinC82}]
If $f:D^n \rightarrow \{0,1\}$ has a word-RAM algorithm with running time $T(n)$ using $S(n)$ $w$-bit words,
then there exists an $n$-variate branching program $\Program$ over $D$ computing $f$, of length $O(T(n))$ and size~$2^{O(wS(n)+ \log n)}$.
\end{lemma}

\paragraph{Embedded Rectangles.}
If $A \subseteq [1,n]$, a point $\tau \in D^A$ (i.e. a function $\tau:A\to D$) is called a \emph{partial input} on $A$. If $\tau_1, \tau_2$ are partial inputs on $A_1, A_2 \in [1,n]$, $A_1 \cap A_2 = \emptyset$, then $\tau_1\tau_2$ is the partial input on $A_1 \cup A_2$ agreeing with $\tau_1$ on $A_1$ and with $\tau_2$ on $A_2$. For sets $B\subseteq D^{[1,n]}$ and $A\subseteq [1,n]$ we define $B_{A}$, the \emph{projection} of $B$ onto $A$, as the set of all partial inputs on $A$ which agree with some input in $B$.
 An embedded rectangle $R$ is a triple $(B, A_1, A_2)$, where $A_1$ and $A_2$ are disjoint subsets of $[1,n]$, and $B\subseteq D^{[1,n]}$ satisfies: (i) $B_{[1,n] \setminus {A_1 \cup A_2}}$ consists of a single partial input $\sigma$,  (ii) if $\tau_1 \in B_{A_1}$, and $\tau_2 \in B_{A_2}$, then $\tau_1\tau_2\sigma \in B$. For an embedded rectangle $R = (B, A_1, A_2)$, and $j \in \{1,2\}$ we define:
\begin{align*}
m_j(R) &= |A_j|  & m(R) &= \min(m_1(R), m_2(R)) \\
\alpha_j(R) &= |B_{A_j}| / |D|^{|A_j|} & \alpha(R) &= \min(\alpha_1(R), \alpha_2(R))
\end{align*}
Given a small branching program $\Program$ it can be shown that $\Program^{-1}(1)$,
the set of all YES-inputs, contains a relatively large embedded rectangle. Namely,

\begin{lemma}[Corollary 5.4 (i)~\cite{DecisionProblemsRandomized}] 
\label{lm:largeR}
	Let $k \ge 8$ be an integer, $q \le 2^{-40}k^{-8}$, $n \ge r \ge q^{-5k^2}$. Let $\Program$ be a $n$-variate branching program over domain $D$ of length at most $(k-2)n$ and size~$2^S$.
	Then there is an embedded rectangle $R$ contained in $\Program^{-1}(1)$ satisfying $m(R) = m_1(R) = m_2(R) \ge q^{2k^2}n/2$ and $\alpha(R) \ge 2^{-q^{1/2}m(R)-Sr}|\Program^{-1}(1)|/|D^n|$.
\end{lemma}

\paragraph{Element Bidistinctness.}
We say that two elements $x = (x_1, x_2)$ and $y = (y_1, y_2)$ of the Cartesian product $D \times D$ are \emph{bidistinct} if both $x_1 \neq y_2$ and $x_2 \neq y_1$. The element bidistinctness function $EB:(D\times D)^{n}\to \{0,1\}$ is defined to be~$1$ iff for every pair
of indices $1\le i,j\le n$ the $i$-th and $j$-th pair are bidistinct. Note that computing EB for $(s_1,t_1),\ldots,(s_{n},t_{n})$ is equivalent
to deciding if $LCS(s_1\ldots s_{n}, t_1\ldots t_{n})\ge 1$. Thus the problem of computing the longest common substring of two strings over $\Sigma =D$ is at least as hard as the EB problem.
Below we show a time-space trade-off lower bound for element bidistinctness.

\begin{lemma}
\label{lm:EBisBIG}
	If $|D| \ge 2n^2$, at least a fraction $1/e$ of inputs belong to $EB^{-1}(1)$.
\end{lemma}
\begin{proof}
	The size of $EB^{-1}(1)$ is at least $(|D|-1)^2\cdot (|D|-2)^2 \cdot \ldots \cdot (|D|-n)^2$. Hence, $|EB^{-1}(1)| = |D|^{2n} \prod\limits_{i=1}^{n}(1 - \frac{i}{|D|})^2 \ge |D|^{2n} (1-\frac{1}{2n})^{2n} \ge |D|^{2n}/{e}$. 
	\qed
\end{proof}

\begin{lemma}
\label{lm:RisSmall}
	For any embedded rectangle $R = (B, A_1, A_2) \subseteq EB^{-1}(1)$ we have $\alpha(R) \leq 2^{-2m(R)}$.
\end{lemma}
\begin{proof}
Let $S_j$ be the subset of $D \times D$ that appear on indices in $A_j$, i.e., $S_j = \bigcup_{\tau \in B_{A_j}}\{\tau(i):i\in A_j\}$, $j=1,2$. Clearly, all elements in $S_1$ must be bidistinct from all elements in $S_2$. If this was not the case $B$ would contain a vector with two non-bidistinct elements of $D \times D$. 
We will prove that $\min(|S_1|,|S_2|) \leq |D|^2/4$. Let us first argue that this implies the lemma.
For $j=1$ or $j=2$, we get that $|B_{A_j}| \leq (|D|^2/4)^{|A_j|}$, and thus $\alpha_j(R) \leq (|D|^2/4)^{|A_j|}/(|D|^2)^{|A_j|} = 4^{-|A_j|} \leq 4^{-m(R)} = 2^{-2m(R)}$.

It remains to prove that $\min(|S_1|,|S_2|) \leq |D|^2/4$. For $j\in\{1,2\}$ let $X_j$ and $Y_j$ denote the set of first and second 
coordinates that appear in $S_j$. Note that by bidistinctness $X_1\cap Y_2 = X_2\cap Y_1 = \emptyset$.
Moreover $|S_j|\le |X_j||Y_j|$ and therefore $\sqrt{|S_j|}\le \sqrt{|X_j||Y_j|} \le \tfrac{1}{2}(|X_j|+|Y_j|)$. 
Consequently
$2(\sqrt{|S_1|}+\sqrt{|S_2|})\le |X_1|+|Y_1|+|X_2|+|Y_2| =  (|X_1|+|Y_2|)+(|Y_1|+|X_2|)\le 2|D|$
and thus $\min(\sqrt{|S_1|},\sqrt{|S_2|})\le |D|/2$, i.e. $\min(|S_1|,|S_2|)\le |D|^2/4$ as claimed.
\qed
\end{proof}

\begin{theorem}
	Any $n$-variate branching program $\Program$ of length $T$ and size $2^S$ over domain $D$, $|D| \ge 2n^2$, which computes the element bidistinctness function $EB$, requires $T = \Omega(n \sqrt{\log(n/S) / \log \log (n/S)})$ time.
\end{theorem}
\begin{proof}
The proof repeats the proof of Theorem 6.13~\cite{DecisionProblemsRandomized}. 
We restore the details omitted in~\cite{DecisionProblemsRandomized} for the sake of completeness.  
Suppose that the length of $\Program$ is $T = (k-2) n/2$ and size~$2^S$. 
Apply Lemma~\ref{lm:largeR} with $q = 2^{-40} k^{-8}$ and $r =\ceil{q^{-5k^2}}$. 
We then obtain an embedded rectangle $R \in EB^{-1}(1)$ such that $m(R) \ge q^{2k^2}n/4$ and $\alpha(R) \ge 2^{-q^{1/2} m(R) - Sr}/e = 2^{-q^{1/2}m(R) - Sr - \log e}$.
From Lemma~\ref{lm:RisSmall} we have $2^{-2m(R)} \ge 2^{-q^{1/2}m(R) - Sr - \log e}$ and thus $Sr \ge m(R)(2-q^{1/2}) - \log e \ge m(R)/2$. Consequently, $S \ge 	
q^{2k^2}n/(8r)$. Remember that $q = 2^{-40} k^{-8}$ and $r =\ceil{q^{-5k^2}}$, which means that $\Program$ requires at least $k^{-ck^2} n$ space for 
some constant $c>0$. That is, $k^{ck^2} \ge n/S$, which implies $k = \Omega(\sqrt{\log(n/S)/\log\log(n/S})$. Substituting $k =2T/n+2$, we obtain the 
claimed bound.
\qed
\end{proof}

\begin{corollary}
Any deterministic RAM algorithm that solves the element bidistinctness (EB) problem on inputs in $(D \times D)^{n}$, $|D| \ge 2n^2$, using $\tau\le \frac{n}{\log n}$ space, must use at least $\Omega (n \sqrt{\log (n /(\tau \log n))/ \log \log (n/(\tau \log n))}\big)$ time.
\end{corollary}

\begin{corollary}[\autoref{thm:lowerbound}]
Given two documents of total length $n$ from an alphabet $\Sigma$ of size at least $n^2$, any deterministic RAM algorithm, which uses $\tau\le \frac{n}{\log n}$ space to compute the longest common substring of both documents, must use time $\Omega (n \sqrt{\log (n/(\tau\log n))/ \log\log (n/(\tau \log n))})$.
\end{corollary}

\section{Conclusions}
The main problem left open by our work is to settle the optimal time-space product for the LCS problem. While it is tempting to guess that the answer lies in the vicinity of $\Theta(n^2)$, it seems really difficult to substantially improve our lower bound. Strong time-space product lower bounds have so far only been established in weaker models (e.g., the comparison model) or for multi-output problems (e.g., sorting an array, outputting its distinct elements and various pattern matching problems). Proving an $\Omega(n^2)$ time-space product lower bound in the RAM model for \emph{any} problem where the output fits in a constant number of words (e.g., the LCS problem) is a major open problem.

\bibliographystyle{splncs03}
\bibliography{paper}
\end{document}